\documentclass[letter,11pt]{article}

\usepackage[english]{babel}
\usepackage[utf8x]{inputenc}
\usepackage[T1]{fontenc}
\usepackage{amsfonts}

\usepackage[margin = 1in]{geometry}

\usepackage{amsmath}

\usepackage{amsthm}
\usepackage{algpseudocode,algorithm,algorithmicx}
\usepackage{bbm}
\usepackage{enumitem}
\usepackage{graphicx}
\usepackage[colorinlistoftodos]{todonotes}
\usepackage[colorlinks=true, allcolors=blue]{hyperref}
\usepackage{thmtools}
\usepackage{thm-restate}

\newtheorem{theorem}{Theorem}
\newtheorem{lemma}[theorem]{Lemma}

\newtheorem{definition}[theorem]{Definition}

\newtheorem{corollary}[theorem]{Corollary}
\newtheorem{fact}[theorem]{Fact}
\newcommand{\ex}{\mathbb{E}}
\newcommand{\eat}[1]{}

\newcommand{\norm}[1]{||#1||}

\newcommand{\mpq}{\mathcal{M}^p}

\newcommand{\D}[1]{\text{d}#1}

\title{Privately Answering Counting Queries with Generalized Gaussian Mechanisms}
\author{Arun Ganesh
  \footnote{Department of Electrical Engineering and Computer Sciences, UC Berkeley. Email: \texttt{arunganesh@berkeley.edu}.} \and 
  Jiazheng Zhao
  \footnote{Department of Electrical Engineering and Computer Sciences, UC Berkeley. Email: \texttt{zoo@berkeley.edu}.}}
\date{}
\begin{document}
\maketitle
\thispagestyle{empty} 
\setcounter{page}{0}
\begin{abstract}
We consider the problem of answering $k$ counting (i.e. sensitivity-1) queries about a database with $(\epsilon, \delta)$-differential privacy. We give a mechanism such that if the true answers to the queries are the vector $d$, the mechanism outputs answers $\tilde{d}$ with the $\ell_\infty$-error guarantee:

$$\ex\left[\norm{\tilde{d} - d}_\infty\right] = O\left(\frac{\sqrt{k \log \log \log k \log(1/\delta)}}{\epsilon}\right).$$

This reduces the multiplicative gap between the best known upper and lower bounds on $\ell_\infty$-error from $O(\sqrt{\log \log k})$ to $O(\sqrt{\log \log \log k})$. Our main technical contribution is an analysis of the family of mechanisms of the following form for answering counting queries: Sample $x$ from a \textit{Generalized Gaussian}, i.e. with probability proportional to $\exp(-(\norm{x}_p/\sigma)^p)$, and output $\tilde{d} = d + x$. This family of mechanisms offers a tradeoff between $\ell_1$ and $\ell_\infty$-error guarantees and may be of independent interest. For $p = O(\log \log k)$, this mechanism already matches the previous best known $\ell_\infty$-error bound. We arrive at our main result by composing this mechanism for $p = O(\log \log \log k)$ with the sparse vector mechanism, generalizing a technique of Steinke and Ullman. 
\end{abstract}
\newpage

\section{Introduction}

A fundamental question in data analysis is to, given a database, release answers to $k$ numerical queries about a database $d$, balancing the goals of preserving the \textit{privacy} of the individuals whose data comprises the database and preserving the \textit{utility} of the answers to the queries. A standard formal guarantee for privacy is $(\epsilon, \delta)$-differential privacy~\cite{DworkMNS06, DworkKMMN06}. A mechanism $\mathcal{M}$ that takes database $d$ as input and outputs (a distribution over) answers $\tilde{d}$ to the queries is $(\epsilon, \delta)$-differentially private if for any two databases $d, d'$ which differ by only one individual and for any set of outcomes $S$, we have:

\begin{equation}\label{eq:approxdp}
\Pr_{\tilde{d} \sim \mathcal{M}(d)}\left[\tilde{d} \in S\right] \leq e^\epsilon \Pr_{\tilde{d} \sim \mathcal{M}(d')}\left[\tilde{d} \in S\right] + \delta.
\end{equation}

When $\delta = 0$, this property is referred to $\epsilon$-differential privacy. Without loss of generality, we will treat $d$ (resp. $\tilde{d}$) as a $k$-dimensional vector corresponding to the answers to the queries (resp. the answers outputted by the mechanism). In this paper, we focus on the setting of \textit{counting queries}, i.e. queries for which the presence of each individual in the database affects the answers by at most 1. In turn, throughout the paper we say a mechanism taking vectors in $\mathbb{R}^k$ as input and outputting distributions over $\mathbb{R}^k$ is $(\epsilon, \delta)$-differentially private if \eqref{eq:approxdp} holds for any two $k$-dimensional vectors $d, d'$ such that $\norm{d - d'}_\infty \leq 1$ and any subset $S$ of $\mathbb{R}^k$.

To balance the goals of privacy and utility, we seek a mechanism $\mathcal{M}$ that minimizes some objective function of the (distribution of) additive errors $\tilde{d} - d$, while satisfying \eqref{eq:approxdp}. One natural and well-understood objective function is the $\ell_1$-error $\norm{\tilde{d} - d}_1/k$, which gives the average absolute error of the answers to the queries. The well-known and simple \textit{Laplace mechanism}~\cite{DworkMNS06}, which outputs $\tilde{d} = d + x$ with probability proportional to $\exp(-\norm{x}_1 / \sigma)$ for an appropriate value of $\sigma$, achieves expected $\ell_1$-error of $O(\min\{\sqrt{k \log (1/\delta)}, k\}/\epsilon)$. A line of works on lower bounds \cite{HardtT10, BunUV14} culminated in a result of \cite{SteinkeU17} showing this is optimal up to constants.

A less well-understood objective function is the $\ell_\infty$-error $\norm{\tilde{d}-d}_\infty$, which gives the maximum absolute error of the answers to the queries. The maximum absolute error is of course a more strict objective function than the average absolute error; indeed, the Laplace mechanism only achieves error $O(k \log k /\epsilon)$ and the Gaussian mechanism (which outputs $\tilde{d} = d + x$ with probability proportional to $\exp(-\norm{x}_2^2/\sigma^2)$ for an appropriate value of $\sigma$) achieves error $O(\sqrt{k \log k \log (1/\delta)}/\epsilon)$. The first improvements on $\ell_\infty$-error over the Laplace and Gaussian mechanisms were given by \cite{SteinkeU17}\footnote{Their paper considers the problem setting where queries ask what fraction of $n$ individuals satisfy some property, i.e. queries have sensitivity $1/n$ instead of $1$, and the goal is to find the minimum $n$ needed to achieve error at most $\alpha$. Achieving error $\Delta$ with probability $1-\rho$ in our setting is equivalent to needing $n \geq \Delta / \alpha$ to achieve error $\alpha$ with probability $1-\rho$ in their setting.}. To summarize, the results of that paper (which prior to this paper were all the best known results) are:

\begin{itemize}
    \item An $\epsilon$-differentially private mechanism satisfying: 
    \begin{equation}\label{eq:su-pure}
        \Pr_{\tilde{d} \sim \mathcal{M}(d)}\left[\norm{\tilde{d} - d}_\infty \geq O\left(\frac{k}{\epsilon}\right)\right] \leq e^{-\Omega(k)},
    \end{equation}
    (this matches a lower bound of \cite{HardtR10} up to constants).
    \item An $(\epsilon, \delta)$-differentially private mechanism satisfying: 
    \begin{equation}\label{eq:su-apx}
        \Pr_{\tilde{d} \sim \mathcal{M}(d)}\left[\norm{\tilde{d} - d}_\infty \geq O\left(\frac{\sqrt{k \log \log k \log(1/\delta)}}{\epsilon}\right)\right] \leq e^{-\log^{\Omega(1)} k}.
    \end{equation}
    \item A lower bound showing any $(\epsilon, \delta)$-differentially private mechanism must satisfy: 
    \begin{equation}\label{eq:su-lb}
       \ex_{\tilde{d} \sim \mathcal{M}(d)}\left[\norm{\tilde{d} - d}_\infty \right] \geq \Omega\left(\frac{\sqrt{k \log (1/\delta)}}{\epsilon}\right).
    \end{equation}
\end{itemize}

The mechanism achieving \eqref{eq:su-apx} starts by taking the Gaussian mechanism, and then uses the sparse vector mechanism to correct the entries of $x$ with large error in a private manner. In conjunction with \eqref{eq:su-lb}, this nearly but not completely settles the question of the optimal error for this problem. In this paper, we further close the gap between \eqref{eq:su-apx} and \eqref{eq:su-lb} by providing a mechanism that achieves error $O(\sqrt{k \log \log \log k \log(1/\delta)}/\epsilon)$. 

\subsection{Our Results and Techniques}

Our first result is as follows:
\begin{theorem}\label{thm:main-1-simple}
For all  $1 \leq p \leq \log k$, $\epsilon \leq O(1)$, $\delta \in [2^{-O(k / p)}, 1/k]$, there exists a $(\epsilon, \delta)$-differentially private mechanism $\mathcal{M}$ that takes in a vector $d \in \mathbb{R}^k$ and outputs a random $\tilde{d} \in \mathbb{R}^k$ such that for some sufficiently large constant $c$, and all $t \geq 0$:
$$\Pr_{\tilde{d} \sim \mpq_\sigma(d)}\left[\norm{\tilde{d} - d}_\infty \geq \frac{c t \sqrt{kp} \log^{1/p} k \sqrt{\log (1/\delta)}}{\epsilon}\right] \leq e^{-t^p \log k}$$
In particular, this implies:
$$\ex_{\tilde{d} \sim \mathcal{M}(d)}[\norm{\tilde{d} - d}_\infty] = O\left(\frac{\sqrt{kp}\log^{1/p} k \sqrt{\log (1/\delta)}}{\epsilon}\right).$$
We also have for all $1 \leq q \leq p$:
$$\ex_{\tilde{d} \sim \mathcal{M}(d)}\left[\frac{\norm{\tilde{d} - d}_q}{k^{1/q}}\right] = O\left(\frac{\sqrt{kp \log (1/\delta)}}{\epsilon}\right).$$
\end{theorem} 
We note that the lower bound on $\delta$ in Theorem~\ref{thm:main-1-simple} can easily be removed: if $\delta$ is smaller than $2^{-O(k / p)}$, we can instead use the mechanism achieving \eqref{eq:su-pure}, which matches the error guarantees of Theorem~\ref{thm:main-1-simple} in this range of $\delta$. 

Setting $p = \Theta(\log \log k)$, this result matches the error bound of \eqref{eq:su-apx}. However, this result improves on \eqref{eq:su-apx} qualitatively. Although the mechanism achieving \eqref{eq:su-apx} is already not too complex, the mechanism we use is even simpler to describe: We add noise $x$ to $d$ with probability proportional to $\exp(-(\norm{x}_p/\sigma)^p)$ for an appropriate choice of $\sigma$. We refer to this noise distribution as the \textit{Generalized Gaussian with shape $p$ and scale $\sigma$}, as is it referred to in e.g. \cite{Nadarajah05}, and this family of mechanisms as \textit{Generalized Gaussian mechanisms}. Notably, Generalized Gaussian mechanisms retain the property of the Gaussian mechanism that the noise added to each entry of $d$ is independent (unlike the mechanism giving \eqref{eq:su-apx}, which uses dependent noise), and that the noise has a known closed-form distribution that is easy to sample from\footnote{see e.g. \url{https://sccn.ucsd.edu/wiki/Generalized_Gaussian_Probability_Density_Function}.}. To the best of our knowledge, this is the first analysis giving privacy guarantees for Generalized Gaussian mechanisms besides that in \cite{Liu19}. Even then, \cite{Liu19} does not give any closed-form bounds on the value of $\sigma$ needed for privacy in the counting queries setting. This analysis may be of independent interest for other applications where one would normally use the Gaussian mechanism, but may want to use a Generalized Gaussian mechanism with $p > 2$ to trade average-case error guarantees for better worst-case error guarantees.

We give a summary of our analysis here; the full analysis is given in Section~\ref{section:lpp}. We first need to determine what value of $\sigma$ causes the Generalized Gaussian mechanism to be private. Viewing the Generalized Gaussian mechanism as an instance of the exponential mechanism of \cite{McSherryT07}, this reduces to deriving a tail bound on $\norm{x+1}_p^p - \norm{x}_p^p$ for $x$ sampled from the noise distribution. If $p$ is even this is roughly equal to $p \sum_{j=1}^k x_j^{p-1}$. By a Chernoff bound on the signs of each random variable in the sum, this is roughly tail bounded by the sum of $\sqrt{k \log(1/\delta)}$ of the $x_j^{p-1}$ random variables. These variables are distributed according to a \textit{Generalized Gamma} distribution, which we prove is sub-gamma in Section~\ref{section:generalizedgammas}. This gives us the desired tail bound, and thus an upper bound on the $\sigma$ needed to ensure $(\epsilon, \delta)$-differential privacy. To prove the error guarantees, we derive tail bounds on the $\ell_p$-norm of $x$ sampled from Generalized Gaussian distributions, as well as on the coordinates of points sampled from unit-radius $\ell_p$-spheres, the latter of which is done by upper bounding the volume of ``sphere caps'' of these spheres.
 
Building on this result, we give the best-known worst-case error for answering counting queries with $(\epsilon, \delta)$-differential privacy:
\begin{theorem}\label{thm:main-2-simple}
For all $\epsilon \leq O(1)$, $\delta \in [2^{-O(k / \log \log \log k)}, 1/k]$, $t \in [0, O(\frac{\log k}{\log \log k})]$, there exists a $(\epsilon, \delta)$-differentially private mechanism $\mathcal{M}$ that takes in a vector $d \in \mathbb{R}^k$ and outputs a random $\tilde{d} \in \mathbb{R}^k$ such that for a sufficiently large constant $c$:

$$\Pr_{\tilde{d} \sim \mathcal{M}(d)}\left[\norm{\tilde{d} - d}_\infty \geq \frac{c t \sqrt{k \log \log \log k \log (1/\delta)}}{\epsilon}\right] \leq e^{-\log^t k}.$$

In particular, if we choose e.g. $t = 2$ we get:

$$\ex_{\tilde{d} \sim \mathcal{M}(d)}[\norm{\tilde{d} - d}_\infty] = O\left(\frac{\sqrt{k \log \log \log k \log (1/\delta)}}{\epsilon}\right).$$
\end{theorem}

Again, the lower bound on $\delta$ can easily be removed using the mechanism achieving \eqref{eq:su-pure}. This gives an exponential improvement on the multiplicative gap between the upper bound in \eqref{eq:su-apx} and the lower bound in \eqref{eq:su-lb}. We arrive at this result by improving upon Generalized Gaussian mechanisms in the same manner \cite{SteinkeU17} improves upon the Gaussian mechanism: After sampling $x$ from a Generalized Gaussian, we apply the sparse vector mechanism to $x$ to get $\tilde{x}$ which satisfies $\norm{x - \tilde{x}}_\infty \ll \norm{x}_\infty$. We then just output $\tilde{d} = d + x - \tilde{x}$.  The full analysis is given in Section~\ref{section:sv}. Similarly to \cite{SteinkeU17}, the major technical component is showing that at least $k / \log^{\Omega(1)} k$ entries of $x$ are small with high probability, which we do using the tail bounds derived in Section~\ref{section:lpp}. This is necessary for the sparse vector mechanism to satisfy that $\norm{x - \tilde{x}}_\infty$ is, roughly speaking, the $(k/\log^{\Omega(1)} k)$-th largest entry of $x$ rather than the largest entry with high probability.

\subsection{Preliminaries}

For completeness, we restate the noise distribution of interest here:

\begin{definition}
The (multivariate) \textbf{Generalized Gaussian distribution with shape $p$ and scale $\sigma$}, denoted $GGauss(p, \sigma)$, is the distribution over $x \in \mathbb{R}^k$ with probability distribution function (pdf) proportional to $\exp(-(\norm{x}_p/\sigma)^p)$. 
\end{definition}

\subsubsection{Sub-Gamma Random Variables}

The following facts about sub-gamma random variables will be useful in our analysis:

\begin{definition}
A random variable $X$ is \textbf{sub-gamma to the right} with variance $v$ and scale $c$ if:

$$\forall \lambda \in (0, 1/c): \ex[\exp(\lambda(X - \ex[X]))] \leq \exp \left(\frac{\lambda^2 v}{2(1 - c\lambda)}\right).$$

Here, we use the convention $1/c = \infty$ if $c = 0$. We denote the class of such random variables $\Gamma^+(v, c)$. Similarly, a random variable $X$ is \textbf{sub-gamma to the left} with variance $v$ and scale $c$, if $-X \in \Gamma^+(v, c)$, i.e.:

$$\forall \lambda \in (0, 1/c): \ex[\exp(\lambda(\ex[X] - X))] \leq \exp \left(\frac{\lambda^2 v}{2(1 - c\lambda)}\right).$$

We denote the class of such random variables $\Gamma^-(v, c)$.
\end{definition}

We refer the reader to \cite{BoucheronLM13} for a textbook reference for this definition and proofs of the following facts.

\begin{fact}\label{fact:subgammasum}
If for $i \in [n]$ we have a random variable $X_i \in \Gamma^+(v_i, c_i)$, then $X = \sum_{i \in [n]} X_i$ satisfies $X \in \Gamma^+(\sum_{i \in [n]} v_i, \max_{i \in [n]} c_i)$ (and the same relation holds for $\Gamma^-(v, c)$).
\end{fact}

\begin{lemma}\label{lemma:gammatailbound}
If $X \in \Gamma^+(v, c)$ then for all $t> 0$:
$$\Pr[X > \ex[X] + \sqrt{2vt} + ct] \leq e^{-t}.$$
Similarly, if $X \in \Gamma^-(v, c)$ then for all $t> 0$:
$$\Pr[X < \ex[X] - \sqrt{2vt} - ct] \leq e^{-t}.$$
\end{lemma}

\begin{fact}
Let $X \sim Gamma(a)$, i.e. $X$ has pdf satisfying:

$$p(x) \propto x^{a-1}e^{-x}.$$

Then $X$ satisfies $X \in \Gamma^+(a, 1)$ and $X \in \Gamma^-(a, 0)$.
\end{fact}

\subsubsection{Other Probability Facts}
We will use the following standard fact to relate distributions of variables to the distributions of their powers:

\begin{fact}[Change of Variables for Powers]
Let $X$ be distributed over $(0, \infty)$ with pdf proportional to $f(x)$. Let $Y$ be the random variable $X^c$ for $c > 0$. Then $Y$ has pdf proportional to $y^{\frac{1}{c} - 1} f(y^{\frac{1}{c}})$.

\end{fact}

Finally, we'll use the following standard tail bounds:

\begin{lemma}[Laplace Tail Bound]
Let $X$ be a Laplace random variable with scale $b$, $Lap(b)$. That is, $X$ has pdf proportional to $\exp(-|x|/b)$. Then we have $\Pr[|x| \geq tb] \leq e^{-t}$.
\end{lemma}

\begin{lemma}[Chernoff Bound]\label{lemma:chernoffna}
Let $X_1, X_2, \ldots X_k$ be independent Bernoulli random variables. Let $\mu = \ex\left[\sum_{i \in [k]} X_i\right]$. Then for $t \in [0, 1]$, we have:
$$\Pr\left[\sum_{i \in [k]} X_i \geq (1+t) \mu\right] \leq \exp\left(-\frac{t^2 \mu}{3}\right).$$
\end{lemma}
\section{Generalized Gaussian Mechanisms}\label{section:lpp}

In this section, we analyze the Generalized Gaussian mechanism that given database $d$, samples $x \sim GGauss(p, \sigma)$ and outputs $\tilde{d} = d + x$. We denote this mechanism $\mpq_\sigma$. When $p = 1$ this is the Laplace mechanism, and when $p = 2$ this is the Gaussian mechanism. 

\subsection{Privacy Guarantees}
We first determine what $\sigma$ is needed to make this mechanism private. We start with the following lemma, which gives a tail bound on the change in the ``utility'' function $\norm{\tilde{d} - d}_p^p$ if $d$ changes by $\Delta \in [-1, 1]^k$:

\begin{lemma}\label{lemma:lppdifference}
Let $x \in \mathbb{R}^k$ be sampled from $GGauss(p, \sigma)$. Then for $4 \leq p \leq \log k$ that is an even integer, $\delta \in [2^{-O(k / p)}, 1/k]$, and any $\Delta \in [-1, 1]^k$ we have with probability $1 - \delta$, for a sufficiently large constant $c$:

$$\norm{x - \Delta}_p^p - \norm{x}_p^p \leq cp \left[k^{1/p-1/2} \sqrt{p \log (1/\delta)} \norm{x}_p^{p-1} + 2k^{\frac{p}{2}}p\right]$$
\end{lemma}
We remark that the requirement that $p$ be an even integer can be dropped by generalizing the proofs in this section appropriately. However, we can reduce proving Theorem~\ref{thm:main-1-simple} for all $p$ to proving it for only even $p$ by rounding $p$ up to the nearest even integer, and only considering even $p$ simplifies the presentation. So, we stick to considering only even $p$.

\begin{proof}

By symmetry of $GGauss(p, \sigma)$ we can assume $\Delta$ has all negative entries. Then we have:

$$\norm{x - \Delta}_p^p - \norm{x}_p^p = \sum_{i = 1}^{k} ((x_i - \Delta_j)^p - x_i^p) $$
$$= \sum_{i = 1}^{k} \int_{x_i}^{x_i - \Delta} py^{p-1} \D{y} \leq \sum_{i = 1}^{k} \int_{x_i}^{x_i - \Delta} p(x_i - \Delta_i)^{p-1} \D{y} \leq p \sum_{i = 1}^{k} (x_i - \Delta_i)^{p-1} \leq p \sum_{i = 1}^{k} (x_i + 1)^{p-1}.$$

We want to replace the terms $(x_i + 1)^{p-1}$ with terms $x_i^{p-1}$ since the latter's distribution is more easily analyzed. To do so, we use the following observation:

\begin{fact}\label{fact:conversion}
If $p \leq \sqrt{k}/2$:
\begin{itemize}
    \item If $x_i > \sqrt{k}$, then we have $(x_i + 1)^{p-1} \leq \left(1 + \frac{1}{\sqrt{k}}\right)^{p-1} x_j^{p-1} \leq \left(1 + \frac{2p}{\sqrt{k}}\right) x_j^{p-1}$.
    \item If $|x_i| \leq \sqrt{k}$, then we have $(x_i + 1)^{p-1} - x_i^{p-1} \leq (\sqrt{k} + 1)^{p-1} - \sqrt{k}^{p-1} \leq 2k^{\frac{p}{2}-1}p$.
    \item If $x_i < -\sqrt{k}$, then we have $(x_i + 1)^{p-1} \leq \left(1 - \frac{1}{\sqrt{k}}\right)^{p-1} x_j^{p-1} \leq \left(1 - \frac{2p}{\sqrt{k}}\right) x_j^{p-1}$.
\end{itemize}
\end{fact}

Fact~\ref{fact:conversion} gives:

$$\sum_{i = 1}^{k} (x_i+1)^{p-1} \leq \left(1 - \frac{2p}{\sqrt{k}}\right)\sum_{i: x_i < 0} x_i^{p-1} + \left(1 + \frac{2p}{\sqrt{k}}\right)\sum_{i: x_i \geq 0} x_i^{p-1} + 2k^{\frac{p}{2}}p.$$

It now suffices to show that for some sufficiently large constant $c$:

\begin{equation}\label{eq:sufficientcondition}
-\left(1 - \frac{2p}{\sqrt{k}}\right)\sum_{i: x_i < 0} |x_i|^{p-1} + \left(1 + \frac{2p}{\sqrt{k}}\right)\sum_{i: x_i \geq 0} |x_i|^{p-1} \leq ck^{1/p-1/2} \sqrt{p\log (1/\delta)} \norm{x}_p^{p-1},
\end{equation}

with probability at least $1-\delta$. Note that each $x_i$ is sampled independently with probability proportional to $\exp(-(|x_i|/\sigma)^p)$. Since multiplying $x$ by a constant does not affect whether \eqref{eq:sufficientcondition} holds, it suffices to show \eqref{eq:sufficientcondition} when each $x_i$ is independently sampled with probability proportional to $\exp(-|x_i|^p)$, i.e. when $\sigma = 1$.
By change of variables, $y_i = |x_i|^{p-1}$ is sampled from the distribution with pdf proportional to $y_i^{\frac{1}{p-1} - 1} \exp(-y_i^{\frac{p}{p-1}})$. This is the Generalized Gamma random variable with parameters $(\frac{1}{p-1}, \frac{p}{p-1})$, which we denote $GGamma(\frac{1}{p-1}, \frac{p}{p-1})$. We show the following property of this random variable in Appendix~\ref{section:generalizedgammas}:

\begin{lemma}\label{lemma:ggconcentration}
For any $p \geq 4$, let $Y$ be the random variable $GGamma(\frac{1}{p-1}, \frac{p}{p-1})$, let $\mu = \ex[Y]$. Then $\mu \in [1/p, 1.2/p), Y \in \Gamma^+(\mu, 1)$, and $Y \in \Gamma^-(\mu, 3/2)$. 
\end{lemma}

Let $k'$ be the number of positive coordinates in $x$. A Chernoff bound gives that $k' \leq \frac{k}{2} + 3\sqrt{k\log\frac{1}{\delta}}$ with probability $1 - \delta/3$. By Lemma~\ref{lemma:ggconcentration} and Fact~\ref{fact:subgammasum} $\sum_{i: x_i < 0} |x_i|^{p-1}$ is in $\Gamma^-((k-k')\mu, 3/2)$ and $\sum_{i: x_i \geq 0} |x_i|^{p-1}$ is in $\Gamma^+(k'\mu, 1)$ for $\mu$ as defined in Lemma~\ref{lemma:ggconcentration}. We now apply Lemma~\ref{lemma:gammatailbound} with $t = \log(6/\delta)$ to each sum. Since $\delta \geq 2^{-O(k/\sqrt{p})}$, $\log (6/\delta) = O(\sqrt{k \log(1/\delta)/p})$, i.e. we are still in the range of $\delta$ for which the square-root term $\sqrt{2vt}$ in the tail bound of Lemma~\ref{lemma:gammatailbound} is at least the linear term $ct$. So Lemma~\ref{lemma:gammatailbound} combined with the Chernoff bound gives that with probability $1-2\delta/3$ for some sufficiently large constant $c'$:

\begin{align}
&-\left(1 - \frac{2p}{\sqrt{k}}\right)\sum_{i: x_i < 0} |x_i|^{p-1} + \left(1 + \frac{2p}{\sqrt{k}}\right)\sum_{i: x_i \geq 0} |x_i|^{p-1} \nonumber\\
\leq& -\left(1 - \frac{2p}{\sqrt{k}}\right)\left( (k-k')\mu - \frac{c'}{12}\sqrt{k\mu \log(1/\delta)}\right) + \left(1 + \frac{2p}{\sqrt{k}}\right) \left(k'\mu + \frac{c'}{12}\sqrt{k \mu \log(1/\delta)}\right)\nonumber\\ 
\leq& (2k' - k)\mu  + (2\sqrt{k}p)\mu + \frac{c'}{4} \sqrt{k \mu \log(1/\delta)}\nonumber\\ 
\leq& 6\mu \sqrt{k \log (1/\delta)}  + (2\sqrt{k}p)\mu + \frac{c'}{3} \mu \sqrt{k p \log(1/\delta)}\nonumber\\ 
\leq& k\mu \cdot \frac{c'}{3}\left(\frac{\sqrt{\log (1/\delta)}}{\sqrt{k}} + \frac{p}{\sqrt{k}} + \frac{\sqrt{p \log(1/\delta)}}{\sqrt{k}}\right) \leq k \mu \cdot c' \frac{\sqrt{p \log(1/\delta)}}{\sqrt{k}}.\label{eq:differencebound}
\end{align}

In the last step, we use that $p \leq \log k \leq \log(1/\delta)$ for the range of $p, \delta$ we consider. On the other hand, by Fact~\ref{fact:subgammasum} $\sum_{i \in [k]} |x_i|^{p-1} = \norm{x}_{p-1}^{p-1}$ is sampled from a random variable in $\Gamma^-(k \mu, 3/2)$ and thus by Lemma~\ref{lemma:ggconcentration} and Lemma~\ref{lemma:gammatailbound} is at least $k \mu$/2 with probability at least $1 - \delta/3$, i.e. $k\mu \leq 2 \norm{x}_{p-1}^{p-1}$ with probability at least $1-\delta/3$. Combined with \eqref{eq:differencebound} by a union bound we get with probability $1 - \delta$:

$$-\left(1 - \frac{2p}{\sqrt{k}}\right)\sum_{i: x_i < 0} |x_i|^{p-1} + \left(1 + \frac{2p}{\sqrt{k}}\right)\sum_{i: x_i \geq 0} |x_i|^{p-1} \leq 2c' \frac{\sqrt{p \log(1/\delta)}}{\sqrt{k}} \cdot \norm{x}_{p-1}^{p-1}$$

Finally, by the Cauchy-Schwarz inequality for any $a \leq b$ and $k$-dimensional $x$ we have $\norm{x}_a \leq k^{1/a - 1/b}\norm{x}_b$. So, $\norm{x}_{p-1}^{p-1} \leq k^{1/p} \norm{x}_p^{p-1}$, giving \eqref{eq:sufficientcondition} with probability $1-\delta$ as desired. 
\end{proof}

Given Lemma~\ref{lemma:lppdifference}, determining the value of $\sigma$ that makes $\mpq_\sigma$ private is fairly straightforward:

\begin{lemma}\label{lemma:sigmachoice}
Let $\mpq_\sigma$ be the mechanism such that $\mpq_\sigma(d)$ samples $x \in \mathbb{R}^k$ from $x \sim GGauss(p, \sigma)$ and outputs $\tilde{d} = d+x$. For $4 \leq p \leq \log k$ that is an even integer, $\epsilon \leq O(1)$, $\delta \in [2^{-O(k / p)}, 1/k]$, and

$$\sigma  = \Theta\left( \frac{\sqrt{kp \log (1/\delta)}}{\epsilon} \right),$$

$\mpq_\sigma$ is $(\epsilon, \delta)$-differentially private.
\end{lemma}

\begin{proof}

It suffices to show that for any vector $\Delta$ in $[-1, 1]^k$:

$$\Pr_{\tilde{d} \sim \mpq_\sigma(d)} \left [\log\left(\frac{ \Pr[\mpq_\sigma(d) = \tilde{d}]}{\Pr[\mpq_\sigma(d+\Delta) = \tilde{d}]}\right) \leq \epsilon \right] = \Pr_{\tilde{d} \sim \mpq_\sigma(d)}\left[\frac{\norm{x-\Delta}_p^p - \norm{x}_p^p}{\sigma^p} \leq \epsilon \right] \geq 1-\delta.$$

Here, we abuse notation by letting $\Pr$ also denote a likelihood function. By Lemma~\ref{lemma:lppdifference} we now have with probability $1 - \delta/2$ for a sufficiently large constant $c$:

$$\norm{x - \Delta}_p^p - \norm{x}_p^p \leq cp k^{1/p - 1/2}\sqrt{p \log (1/\delta)}\norm{x}_p^{p-1} + c p^2 k^{\frac{p}{2}}.$$

The pdf of the rescaled norm $r = \norm{x}_p/\sigma$ is proportional to $r^{k-1} \exp(-r^p)$ over $(0, \infty)$ (the $r^{k-1}$ appears because the $(k-1)$-dimensional surface area of the $\ell_p$-sphere of radius $r$ is proportional to $r^{k-1}$). Letting $R$ denote $r^p$, the pdf of $R$ is proportional to $R^{\frac{k}{p} - 1} \exp(-R)$ by change of variables, i.e. $R$ is the random variable $Gamma(\frac{k}{p})$. Then by the Gamma tail bound, with probability at least $1 - e^{-.001k/p} > 1 - \delta/2$, $R$ is contained in $[\frac{k}{2p}, \frac{2k}{p}]$, so $\norm{x}_p$ is contained in $[\sigma \left(\frac{k}{2p}\right)^{1/p}, \sigma \left(\frac{2k}{p}\right)^{1/p}]$. Then by a union bound, with probability $1 - \delta$:

$$\frac{\norm{x - \Delta}_p^p - \norm{x}_p^p}{\sigma^p} \leq \frac{2cp^{1/p} \sqrt{kp \log (1/\delta)}}{\sigma} + \frac{c p^2 k^{\frac{p}{2}}}{\sigma^p}.$$

Noting that $n^{1/n}$ is contained within $[1, e^{1/e}]$ for all $n \geq 1$, letting 

$$\sigma = \Theta\left( \max\left\{\frac{\sqrt{kp \log (1/\delta)}}{\epsilon}, \frac{\sqrt{k} }{\epsilon^{1/p}}\right\} \right) = \Theta\left( \frac{\sqrt{kp \log (1/\delta)}}{\epsilon} \right),$$

we get that $\frac{\norm{x - \Delta}_p^p - \norm{x}_p^p }{\sigma^p} \leq \epsilon$ with probability $1-\delta$ as desired.
\end{proof}

\subsection{Error Guarantees}

In this section, we analyze the $\ell_\infty$ error of $\mpq_\sigma$, for a given choice of $\delta$ in the range specified in Lemma~\ref{lemma:sigmachoice}. We give an expected error bound, and also a tail bound on the error. The error analysis follows almost immediately from the following lemma, which bounds the fraction of a sphere cap's volume with a large first coordinate:

\begin{lemma}\label{lemma:spherecap}
Let $x$ be chosen uniformly at random from a $k$-dimensional $\ell_p$-sphere with arbitrary radius, i.e. the set of points with $\norm{x}_p = R$ for some $R$, for $p \geq 1$. Then we have:
$$\Pr[|x_1| \geq r \norm{x}_p] \leq \left(1 - r^p\right)^{(k-1)/p} \leq \exp\left(-\frac{(k-1)r^p}{p}\right)$$
\end{lemma}

This lemma or one providing a similar bound likely already exists in the literature, but we are unaware of a reference for it. So, for completeness we give the full proof in Section~\ref{section:deferred-proofs}.

\begin{corollary}\label{corollary:spherecap2}
Let $x$ be chosen uniformly at random from a $k$-dimensional $\ell_p$-sphere with arbitrary radius for $p \geq 1$. Then we have:
$$\Pr[\norm{x}_\infty \geq r \norm{x}_p] \leq k \cdot \exp\left(-\frac{(k-1)r^p}{p}\right)$$
\end{corollary}
\begin{proof}
This follows from Lemma~\ref{lemma:spherecap} and a union bound over all $k$ coordinates (which have identical marginal distributions).
\end{proof}

Combining this corollary with Lemma~\ref{lemma:sigmachoice}, it is fairly straightforward to prove our first main result:

\begin{theorem}\label{thm:main-1}
Let $\mpq_\sigma$ be the mechanism such that $\mpq_\sigma(d)$ samples $x \in \mathbb{R}^k$ from $GGauss(p, \sigma)$ and outputs $\tilde{d} = d+x$. For $4 \leq p \leq \log k$ that is an even integer, For $\epsilon \leq O(1)$, $\delta \in [2^{-O(k / p)}, 1/k]$, and

$$\sigma  = \Theta\left( \frac{\sqrt{kp \log (1/\delta)}}{\epsilon} \right),$$

$\mpq_\sigma$ is $(\epsilon, \delta)$-differentially private and for some sufficiently large constant $c$, and all $t \geq 0$:
$$\Pr_{\tilde{d} \sim \mpq_\sigma(d)}\left[\norm{\tilde{d} - d}_\infty \geq \frac{c t \sqrt{kp} \log^{1/p} k \sqrt{\log (1/\delta)}}{\epsilon}\right] \leq e^{-t^p \log k} + e^{-.001 k / p} $$
\end{theorem}
\begin{proof}
The privacy guarantee follows from Lemma~\ref{lemma:sigmachoice}. 

For the tail bound, if $\norm{\tilde{d}-d}_\infty > \frac{c t \sqrt{k} \log^{1/p} k \sqrt{p \log (1/\delta)}}{\epsilon}$ we have either $\norm{x}_p \geq \Omega\left(\frac{k^{1/2 + 1/p} \sqrt{p \log (1/\delta)}}{\epsilon}\right)$ or $\norm{x}_\infty > \frac{4t \log^{1/p} k}{k^{1/p}} \norm{x}_p$. Recall that $(\norm{x}_p/\sigma)^p$ is distributed according to a $Gamma(\frac{k}{p})$ random variable, and thus by a Gamma tail bound exceeds $2k/p$ with probability at most $e^{-.001 k / p}$. In turn, $\norm{x}_p \geq \left(\frac{2k}{p}\right)^{1/p}\sigma = \Omega\left(\frac{k^{1/2 + 1/p} \sqrt{p \log (1/\delta)}}{\epsilon}\right)$ with at most this probability. Then it follows by setting $r = \frac{4t \log^{1/p} k}{k^{1/p}}$ in Corollary~\ref{corollary:spherecap2} and a union bound that:

$$\Pr\left[\norm{\tilde{d} - d}_\infty \geq \frac{c t \sqrt{k} \log^{1/p} k \sqrt{p \log (1/\delta)}}{\epsilon}\right] \leq \Pr\left[\norm{x}_\infty \geq \frac{4t \log^{1/p} k}{k^{1/p}} \norm{x}_p\right] + e^{-.001 k / p} \leq$$
$$\exp\left( - \frac{(k-1)4^p t^p \log k}{kp}\right) + e^{-.001 k / p} \leq e^{-t^p \log k} + e^{-.001 k / p}.$$
\end{proof}

This proves Theorem~\ref{thm:main-1-simple}, up to some details which we defer to Section~\ref{section:deferred-proofs}.
\section{Composition with Sparse Vector}\label{section:sv}

In this section, we generalize the mechanism of \cite{SteinkeU17}, which is a composition of the Gaussian mechanism and sparse vector mechanism of \cite{DworkNRRV09}, by analyzing a composition of $\mpq_\sigma$ and the sparse vector mechanism instead. The guarantees given by sparse vector can be given in the following form that we will use:

\begin{theorem}[Sparse Vector]\label{thm:sv}
For every $k \geq 1, c_{SV} \leq k, \epsilon_{SV}, \delta_{SV}, \beta_{SV} > 0$, and 

$$\alpha_{SV} \geq O\left(\frac{\sqrt{c_{SV} \log (1/\delta_{SV})} \log (k / \beta_{SV})}{\epsilon_{SV}}\right),$$
there exists a mechanism $SV$ that takes as input $d \in \mathbb{R}^k$ and outputs $\tilde{d} \in \mathbb{R}^k$ such that:
\begin{itemize}
    \item $SV$ is $(\epsilon_{SV}, \delta_{SV})$-differentially private.
    \item If at most $c_{SV}$ entries of $d$ have absolute value strictly greater than $\alpha_{SV} / 2$, then:
    $$\Pr_{\tilde{d} \sim SV(d)}\left[\norm{\tilde{d} - d}_\infty \geq \alpha_{SV}\right] \leq \beta_{SV}.$$
    \item Regardless of the value of $d$ we have for all $t \geq 0$:
    $$\Pr_{\tilde{d} \sim SV(d)}[\norm{\tilde{d} - d} \geq \max\{\norm{d}_\infty, t \sqrt{k \log(1/\delta_{SV})}/\epsilon_{SV})] \leq ke^{- \Omega(t)}.$$
\end{itemize}
\end{theorem}
\begin{proof}
The mechanism is given by modifying the NumericSparse algorithm given as Algorithm 3 in \cite{DworkR14} by outputting $0$ instead of $\bot$ or $0$ for all remaining queries instead of halting prematurely. The first two properties follow from the associated proofs in that text. 

The third property follows because for all entries of $\tilde{d}$ that $SV$ does not output as $0$ (for which the error, i.e. corresponding entry of $\tilde{d} - d$, is of course bounded by $\norm{d}_\infty$), the error  is drawn from $Lap(b)$ where $b = O(\sqrt{k \log(1/\delta_{SV})}/\epsilon_{SV})$. So the maximum error for these (at most $c_{SV} \leq k$) entries is stochastically dominated by the maximum of the absolute value of $k$ of these Laplace random variables, which is at most $tb$ with probability $ke^{-t}$. 
\end{proof}

We now prove our main result:
\begin{theorem}\label{thm:main-2}
For any $4 \leq p \leq \log k$ that is an even integer, $\epsilon \leq O(1)$, $\delta \in [2^{-O(k / p)}, 1/k]$, and $t \in [0, O(\frac{\log k}{\log \log k})]$, there exists a $(\epsilon, \delta)$-differentially private mechanism $\mathcal{M}$ that takes in a vector $d \in \mathbb{R}^k$ and outputs a random $\tilde{d} \in \mathbb{R}^k$ such that for a sufficiently large constant $c$ :

$$\Pr_{\tilde{d} \sim \mathcal{M}(d)}\left[\norm{\tilde{d} - d}_\infty \geq \frac{c t \sqrt{k p \log (1/\delta)} (\log \log k)^{1/p}}{\epsilon}\right] \leq e^{-\log^t k}.$$
\end{theorem}
\begin{proof}
The mechanism is as follows: We sample $x \sim GGauss(p, \sigma)$ for 

$$\sigma  = \Theta\left( \frac{\sqrt{kp \log (1/\delta)}}{\epsilon} \right),$$

If $\norm{x}_p^p > 2k\sigma^p/p$, we output $d$. Otherwise, we instantiate $SV$ from Theorem~\ref{thm:sv} with parameters:

$$\alpha_{SV} = 12t(\log \log k)^{1/p} \sigma \leq  \frac{ct \sqrt{kp \log(1/\delta)} (\log \log k)^{1/p}}{\epsilon}, \qquad c_{SV} = 4k/\log^{2+2t} k,$$
$$\epsilon_{SV} = \epsilon/2,\qquad \delta_{SV} = \delta/3, \qquad \beta_{SV} = \exp(-\log^t k) / 2.$$
We input $x$ to $SV$ to sample $\hat{x}$, and then output $\tilde{d} = d + x - \hat{x}$.

First, note that:

$$\frac{\sqrt{c_{SV} \log (1/\delta_{SV})} \log (k / \beta_{SV})}{\epsilon_{SV}} \leq \frac{\sqrt{\frac{16k}{\log^{2+2t}k} \log (1/\delta)} (\log k + \log^t k)}{\epsilon} \leq \frac{4 \sqrt{k \log(1/\delta)}}{\epsilon},$$

i.e. $\alpha$ satisfies the requirements of Theorem~\ref{thm:sv} as long as the constant hidden in the $\Theta(\cdot)$ notation in the choice of $\sigma$ is sufficiently large.

To analyze the privacy guarantee, this is the composition of:
\begin{itemize}
    \item The mechanism of Theorem~\ref{thm:main-1}, which if the constant hidden in the $\Theta(\cdot)$ in the expression for $\sigma$ is sufficiently large, is $(\epsilon/2, \delta/3)$-differentially private.
    \item The $SV$ mechanism of Theorem~\ref{thm:sv}, with parameters set so it is $(\epsilon/2, \delta/3)$-differentially private.
    \item The event that $\norm{x}_p^p > 2k\sigma^p/p$, causing us to release the database, which we recall from the Proof of Theorem~\ref{thm:main-1} happens with probability at most $2^{-\Omega(k/p)} \leq \delta/3$.
\end{itemize}

By composition, we get that the mechanism is $(\epsilon, \delta)$-differentially private as desired.

To show the tail bound on $\ell_\infty$-error: If $\norm{x}_p^p > 2k\sigma^p/p$, then we have $\tilde{d} = d$, so trivially the tail bound is satisfied. So, it suffices to show that conditional on $\norm{x}_p^p \leq 2k\sigma^p/p$ occurring, we have the tail bound. By a union bound, the guarantees of Theorem~\ref{thm:sv} give that $\norm{\tilde{d} - d}_\infty = \norm{x - \hat{x}}_\infty \leq \alpha_{SV}$ (i.e the tail bound is satisfied) if at most $4k / \log^{2+2t} k$ entries of $x$ have absolute value greater than $\alpha_{SV}/2$ with probability less than, say, $e^{-2 \log^t k}$. Using $r = 3 t \frac{(\log \log k)^{1/p}}{k^{1/p}}$ in Lemma~\ref{lemma:spherecap} and a union bound with the $1-\delta/3$ probability event that $\norm{x}_p \leq (2k/p)^{1/p}\sigma$, for each coordinate $x_i$ of $x$ we have:

$$|x_i| \geq \alpha_{SV}/2 = 6t (\log \log k)^{1/p} \sigma  = 2rk^{1/p} \sigma \geq r \norm{x}_p,$$

with probability at most $\frac{1}{\log^{2+2t} k} + 2^{-\Omega(k/p)}  \leq \frac{2}{\log^{2+2t} k}$. Since we sample $x$ with probability proportional to $\exp(-\sum_{i \in [k]} |x_i|^p / \sigma^p)$, each coordinate's distribution is independent, so using a Chernoff bound we conclude that with probability $e^{-\Omega(k / \log^{2+2t} k)} \leq e^{-2 \log^t k}$ at most $4k / \log^{2+2t} k$  coordinates have absolute value greater than $\alpha_{SV}$ as desired. 
\end{proof}
This proves Theorem~\ref{thm:main-2-simple}, up to some details which we defer to Section~\ref{section:deferred-proofs}.
\section*{Acknowledgements}
The inspiration for this project was a suggestion by Kunal Talwar that Generalized Gaussians could achieve the same asymptotic worst-case errors for query response as the mechanism of Steinke and Ullman. In particular, he suggested a proof sketch of a statement similar to Lemma~\ref{lemma:sigmachoice} which was the basis for our proof that lemma.

\bibliographystyle{alpha}
\bibliography{ref}

\appendix
\section{Deferred Proofs}\label{section:deferred-proofs}
\subsection{Proof of Lemma~\ref{lemma:spherecap}}
 To prove this lemma we'll need the following lemma about convex bodies. 

\begin{lemma}\label{lemma:convexbodies}
Let $A \subseteq B \subset \mathbb{R}^k$ be two compact convex bodies with $A$ contained in $B$, and $A', B'$ be their respective boundaries. Then $\text{Vol}_{k-1}(A') \leq \text{Vol}_{k-1}(B')$, where $\text{Vol}_{k-1}$ denotes the $(k-1)$-dimensional volume.
\end{lemma}

\begin{proof}
For any compact convex body $S$ and its boundary $S'$, the $(k-1)$-dimensional volume of $S'$ satisfies:

$$\text{Vol}_{k-1}(S') \propto \int_{\mathbb{S}^{k}} \text{Vol}_{k-1} (\pi_{\theta^\top} S) \D{\theta},$$

Where $\mathbb{S}^{k}$ is the $k$-dimensional unit sphere and $\pi_{\theta^\top}S$ is the orthogonal projection of $S$ onto the subspace of $\mathbb{R}^k$ orthogonal to $\theta$ (see e.g. Section 5.5 of \cite{KlainRdB97} for a proof of this fact). Since $A \subseteq B$ it follows that for all $\theta$ we have $\text{Vol}_{k-1} (\pi_{\theta^\top} A) \leq \text{Vol}_{k-1} (\pi_{\theta^\top} B)$ and so $\text{Vol}_{k-1}(A') \leq \text{Vol}_{k-1}(B')$.
\end{proof}

The idea behind the proof of Lemma~\ref{lemma:spherecap} is to show that the region of the $\ell_p$-ball with large positive first coordinate is contained within a smaller $\ell_p$-ball, and then apply Lemma~\ref{lemma:convexbodies}:

\begin{proof}[Proof of Lemma~\ref{lemma:spherecap}]
By rescaling, we can assume $\norm{x}_p = 1$ and instead show:
$$\Pr[|x_1| \geq r] \leq \left(1 - r^p\right)^{(k-1)/p}$$

$$\Pr[|x_1| \geq r] = \frac{\text{Vol}_{k-1}\left(\{x: |x_1| \geq r, \norm{x}_p = 1\} \right)}{\text{Vol}_{k-1}\left(x: \norm{x}_p = 1\right)}= \frac{\text{Vol}_{k-1}\left(\{x: x_1 \geq r, \norm{x}_p = 1\} \right)}{\text{Vol}_{k-1}\left(\{x: x_1 \geq 0, \norm{x}_p = 1\}\right)},$$ 

Where $\text{Vol}_{k-1}$ denotes the $(k-1)$-dimensional volume. To bound this ratio, let $v$ be the vector $(r, 0, 0, \ldots, 0)$, and consider the (compact, convex) body $B_1 = \{x: x_1 \geq r, \norm{x - v}_p \leq (1 - r^p)^{1/p}\}$. We have $r^p + (v-r)^p \leq v^p$ for $0 \leq r \leq v$, so $B_1$ contains the (also compact, convex) body $B_2 = \{x: x_1 \geq r, \norm{x}_p \leq 1\}$. Then by Lemma~\ref{lemma:convexbodies} the $(k-1)$-dimensional surface area of $B_1$ is larger than that of $B_2$. The boundary of $B_1$ is the union of the bodies $B_{1,a} := \{x: x_1 = r, \norm{x - v}_p \leq (1 - r^p)^{1/p}\} $ and $B_{1,b} := \{x: x_1 \geq r, \norm{x - v}_p = (1-r^p)^{1/p}\}$, whose intersection has $(k-1)$-dimensional volume 0. Similarly, the boundary of $B_2$ is the union of the bodies $B_{2,a} := \{x: x_1 = r, \norm{x}_p \leq 1\}$ and $B_{2,b} := \{x: x_1 \geq r, \norm{x}_p = 1\}$, whose intersection has $(k-1)$-dimensional volume $0$. See Figure~\ref{fig:spherecap} for an example of a picture of all of these bodies.

\begin{figure}
    \centering
    \includegraphics[width=.5\textwidth]{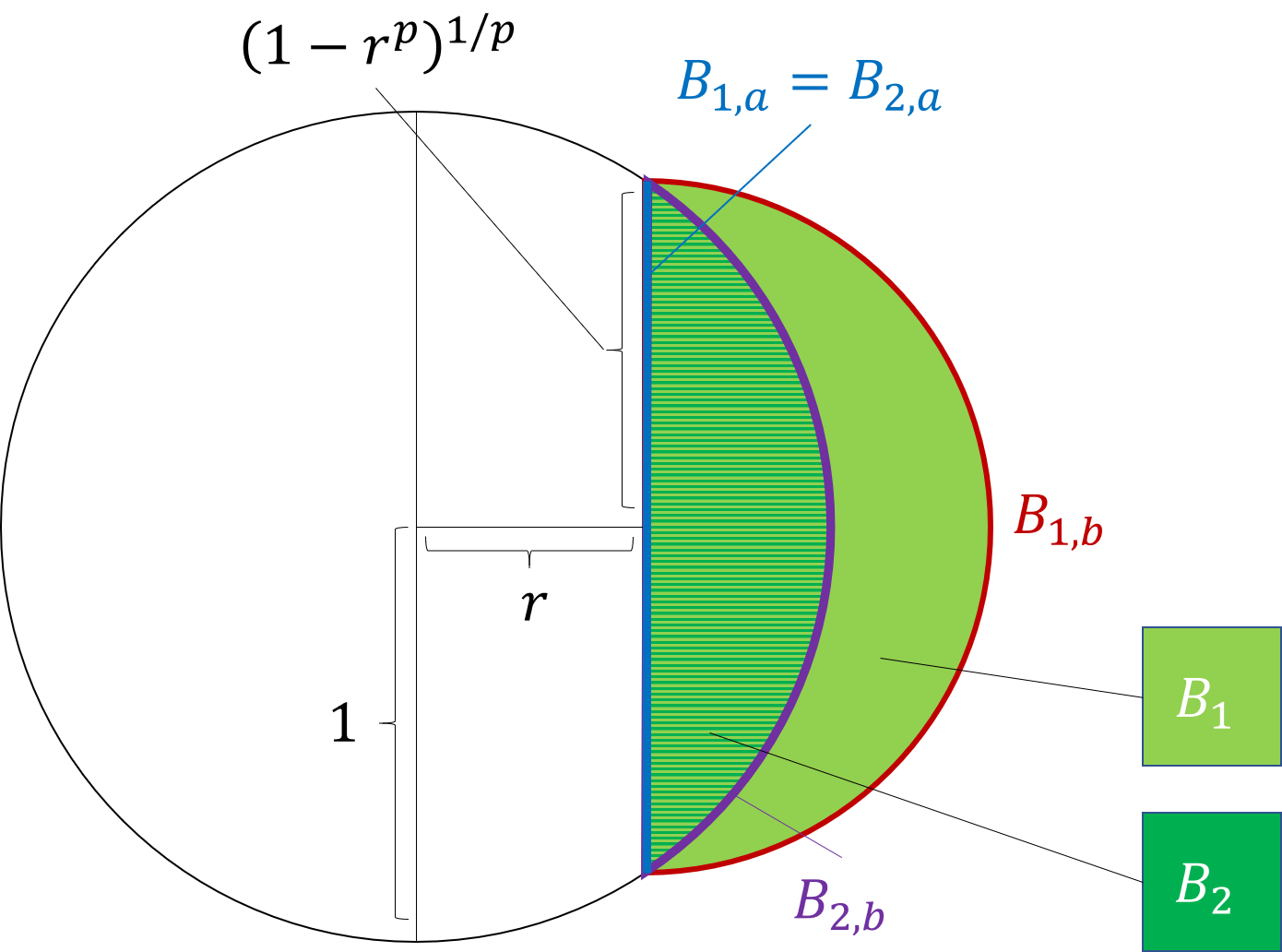}
    \caption{A picture of the bodies in the proof of Lemma~\ref{lemma:spherecap} for $p = 2, k = 2$. $B_2$ has stripes that are the same color as $B_1 \setminus B_2$ to emphasize that $B_1$ contains $B_2$.}
    \label{fig:spherecap}
\end{figure}

Nothing that $B_{1,a} = B_{2,a}$, we conclude that $\text{Vol}_{k-1}(B_{1,b}) \geq \text{Vol}_{k-1}(B_{2,b})$. Now we have:

$$ \frac{\text{Vol}_{k-1}\left(\{x: x_1 \geq r, \norm{x}_p = 1\} \right)}{\text{Vol}_{k-1}\left(\{x: x_1 \geq 0, \norm{x}_p = 1\}\right)} \leq \frac{\text{Vol}_{k-1}(\{x: x_1 \geq r, \norm{x - v}_p = (1-r^p)^{1/p}\})}{\text{Vol}_{k-1}\left(\{x: x_1 \geq 0, \norm{x}_p = 1\}\right)}.$$

The body in the numerator of the final expression is the body in the denominator, but shifted by $v$ and rescaled by $(1-r^p)^{1/p}$ in every dimension. So, the final ratio is at most $(1-r^p)^{(k-1)/p}$.
\end{proof}

\subsection{Proof of Theorem~\ref{thm:main-1-simple}}
We first need the following corollary of Lemma~\ref{lemma:spherecap}:

\begin{corollary}\label{corollary:spherecap3}
Let $x$ be chosen uniformly at random from a $k$-dimensional $\ell_p$-sphere with arbitrary radius for $p \geq 1$. Then we have:
$$\ex[\norm{x}_\infty] \leq \frac{5 \log^{1/p} k}{k^{1/p}} \norm{x}_p$$
\end{corollary}
\begin{proof}
Since $\norm{x}_\infty/\norm{x}_p$ takes values in $[0, 1]$, by Lemma~\ref{lemma:spherecap} we have:

\begin{align*}
\ex[\norm{x}_\infty / \norm{x}_p] &= \int_0^1 \Pr[\norm{x}_\infty / \norm{x}_p \geq r] \D{r}\\
&\leq \int_0^{\frac{2^{1+1/p} \log^{1/p} k}{k^{1/p}}} 1 \D{r} + \int_{\frac{2^{1+1/p} \log^{1/p} k}{k^{1/p}}}^1  k \cdot \exp\left(-\frac{(k-1)r^p}{p}\right) \D{r}\\
&\leq \frac{2^{1+1/p} \log^{1/p} k}{k^{1/p}} + \int_{\frac{2^{1+1/p} \log^{1/p} k}{k^{1/p}}}^1  k \cdot \exp\left(-\frac{(k-1)2^{p+1} \log k}{kp}\right) \D{r}\\
&\leq \frac{2^{1+1/p} \log^{1/p} k}{k^{1/p}} + \int_{\frac{2^{1+1/p} \log^{1/p} k}{k^{1/p}}}^1  k \cdot \exp\left(-2 \log k\right) \D{r}\\
&\leq \frac{2^{1+1/p} \log^{1/p} k}{k^{1/p}} + \frac{1}{k}\\
&\leq \frac{5 \log^{1/p} k}{k^{1/p}}.
\end{align*}

Here we use that $2^p \geq p$ for all $p \geq 1$ and that $(1-\frac{c}{x})^x \leq e^{-c}$ for all $c \geq 0$.
\end{proof}

\begin{proof}[Proof of Theorem~\ref{thm:main-1-simple}]
We use Theorem~\ref{thm:main-1} after rounding $p$ up to the nearest even integer (this loses at most a multiplicative constant in the resulting error bounds). If the constant hidden in $\Theta(\log \log k)$ is a sufficiently large function of $c_1$, this gives the desired tail bound, up to the additive $e^{-.001 k / p}$ in the probability bound (which may be larger than the $e^{-t^p \log k}$ term for large values of $p$). To remove the additive $e^{-.001 k / p}$: if the less than $e^{-.001 k / p} \leq \delta$ probability event that $(\norm{x}_p/\sigma)^p$ exceeds $2k/p$ occurs, we can instead just output $\tilde{d} = d$, i.e. instead set $x = 0$. This gives an $(\epsilon, 2\delta)$-private mechanism that always satisfies $(\norm{x}_p/\sigma)^p \leq 2k/p$, and then we can rescale our choice of $\delta$ appropriately. The tail bound can now be derived as in the proof of Theorem~\ref{thm:main-1}. Similarly, since we always have $(\norm{x}_p/\sigma)^p \leq 2k/p$, the expectation of $\norm{x}_\infty$ follows from Corollary~\ref{corollary:spherecap3}. Finally, the expectation of $\norm{x}_q$ for $1 \leq q \leq p$ follows by using Jensen's inequality twice and the unconditional upper bound on $\norm{x}_p^p$:

$$\ex[\norm{x}_q] \leq \ex[\norm{x}_q^q]^{1/q} = k^{1/q} \ex[|x_1|^q]^{1/q} \leq k^{1/q} \ex[|x_1|^p]^{1/p} = k^{1/q - 1/p} \ex[\norm{x}_p^p]$$
$$\leq  k^{1/q-1/p} \cdot (2k/p)^{1/p} \sigma = O(k^{1/q} \sigma).$$
\end{proof}

\subsection{Proof of Theorem~\ref{thm:main-2-simple}}

\begin{proof}[Proof of Theorem~\ref{thm:main-2-simple}]
The tail bound in Theorem~\ref{thm:main-2-simple} follows immediately from Theorem~\ref{thm:main-2} by choosing $p$ to be an even integer satisfying $p = \Theta(\log \log \log k)$. 

For the expectation, we use the tail bound of Theorem~\ref{thm:main-2-simple}. We have:

$$\ex_{\tilde{d} \sim \mathcal{M}(d)}\left[\norm{\tilde{d} - d}_\infty \right] = \int_0^\infty \Pr[\norm{\tilde{d} - d}_\infty \geq s] \D{s}$$
$$= \int_0^a \Pr[\norm{\tilde{d} - d}_\infty \geq s] \D{s} + \int_a^b \Pr[\norm{\tilde{d} - d}_\infty \geq s] \D{s} + \int_b^\infty \Pr[\norm{\tilde{d} - d}_\infty \geq s] \D{s}.$$

We choose $a = \frac{2c \sqrt{k \log \log \log k \log (1/\delta)}}{\epsilon}$, $b = \frac{k \sqrt{\log (1/\delta)}}{\epsilon}$. The integral over $[0, a]$ is of course bounded by $a$. By Theorem~\ref{thm:main-2}, the integral over $[a, b]$ is bounded by $b \cdot e^{-\log^2 k} \leq \frac{\sqrt{\log(1/\delta)}}{\epsilon} \leq a$.  Finally, to bound the third term, recall that the mechanism of Theorem~\ref{thm:main-2} outputs $d$ (i.e. effectively chooses $x, \hat{x} = 0$ instead) if $\norm{x}_p$ is too large. So, unconditionally we have:

$$\norm{x}_\infty \leq \norm{x}_p \leq (2k/p)^{1/p} \sigma \leq \frac{2c \sqrt{k \log \log \log k \log(1/\delta)} }{\epsilon} \leq b.$$

So by the third property in Theorem~\ref{thm:sv} we have for $s \in [b, \infty)$:

$$\Pr_{\tilde{d}\sim \mathcal{M}(d)}[\norm{\tilde{d} - d}_\infty \geq s] = \Pr_{x, \hat{x}}[\norm{x - \hat{x}}_\infty \geq s] \leq k e^{-\Omega(s / (\sqrt{k \log(1/\delta)} / \epsilon))}.$$

And so by change of variables, with $s' = s / (\sqrt{k \log(1/\delta)} / \epsilon)$:
$$\int_b^\infty \Pr[\norm{\tilde{d} - d}_\infty \geq s] \D{s} \leq \frac{\sqrt{k \log(1/\delta)}}{\epsilon} \int_{\sqrt{k}}^\infty ke^{-\Omega(s')} \D{s'} \leq \frac{k^{1.5} \sqrt{\log(1/\delta)}}{\epsilon} \cdot e^{-\Omega(\sqrt{k})} \leq a.$$

So we conclude 

$$\ex_{\tilde{d} \sim \mathcal{M}(d)}\left[\norm{\tilde{d} - d}_\infty\right] \leq 3a = O\left(\frac{\sqrt{k \log \log \log k \log (1/\delta)}}{\epsilon}\right),$$

as desired.
\end{proof} 
\section{Concentration of Generalized Gammas}\label{section:generalizedgammas}

In this section we consider the Generalized Gamma random variable $GGamma(a, b)$ parameterized by $a, b$ with pdf:

$$p(x) = \frac{b x^{a-1} e^{-x^b}}{\Gamma(a/b)}, x \in (0, \infty).$$

Where the Gamma function $\Gamma(x)$ is defined over the positive reals as

$$\Gamma(z) = \int_0^\infty x^{z-1} e^{-x} \D{x}.$$

We recall that $\Gamma(z)$ is a continuous analog of the factorial in that it satisfies $\Gamma(x+1) = x \cdot \Gamma(x)$. When $b = 1$, $GGamma(a, b)$ is exactly the Gamma random variable $Gamma(a)$ (we will use $Gamma$ to denote the random variable and $\Gamma$ to denote the function to avoid ambiguous notation). 

We want to show that sums of $GGamma(\frac{1}{p-1}, \frac{p}{p-1})$ random variables concentrate nicely. To do this, we will show that they are sub-gamma:

To show that $GGamma(\frac{1}{p-1}, \frac{p}{p-1})$ are sub-gamma, we will relate the moment-generating function of $GGamma(\frac{1}{p-1}, \frac{p}{p-1})$ to that of the Gamma random variable with the same mean using the following facts:

\begin{fact}
For a Generalized Gamma random variable $X \sim GGamma(a, b)$ the moments are $\ex[X^r] = \frac{\Gamma((a+r)/b)}{\Gamma(a/b)}$. In particular, for a Gamma random variable $X \sim Gamma(a)$ the moments are $\ex[X^r] = \frac{\Gamma(a+r)}{\Gamma(a)}$.
\end{fact}

See e.g. Section 17.8.7 of \cite{JohnsonKB95} for a derivation of this fact. Note here that $GGamma(\frac{1}{p-1}, \frac{p}{p-1})$ has mean $\mu = 1/\Gamma(1/p)$. To relate the moments of Generalized Gamma random variables to Gamma random variables' we note the following about $\mu$:

\begin{fact}\label{fact:ggmean}
For all $p \geq 2$, we have $\frac{1}{p} \leq \frac{1}{\Gamma(1/p)} \leq \frac{1.2}{p}$.
\end{fact}

Putting it all together, we get the following lemmas, which combined with Fact~\ref{fact:ggmean} give us Lemma~\ref{lemma:ggconcentration}:

\begin{lemma}
Let $Y = GGamma(\frac{1}{p-1}, \frac{p}{p-1})$ for $p \geq 2$. Then, for $\mu = \ex[Y] = \frac{1}{\Gamma(1/p)}$, we have $Y \in \Gamma^+(\mu, 1)$.
\end{lemma}
\begin{proof}
We compare the moment-generating function of (the centered version of) $Y$ to that of $X = Gamma(\mu)$ where $\mu = \ex[Y]$. $X$ is in $\Gamma(\mu, 1)$ so it suffices to show $Y$'s moment generating function is smaller than $X$'s. First, looking at the moment generating function of $Y$, we have:

\begin{align*}
\ex[e^{\lambda Y}]&=1 + \lambda \mu + \sum_{r=2}^\infty \left[\frac{\lambda^r}{r!} \ex[Y^r] \right]\\
&=1 + \lambda \mu + \sum_{r=2}^\infty \left[\frac{\lambda^r}{r!} \frac{\Gamma(\frac{1}{p} + \frac{r(p-1)}{p})}{\Gamma(\frac{1}{p})} \right]\\
&\stackrel{(a)}{\leq}1 + \lambda \mu + \sum_{r=2}^\infty \left[\frac{\lambda^r}{r!} \frac{\Gamma(\frac{1}{p} + r)}{\Gamma(\frac{1}{p})} \right] \\
&\stackrel{(b)}{\leq}1 + \lambda \mu + \sum_{r=2}^\infty \left[\frac{\lambda^r}{r!} \frac{\Gamma(\mu + r)}{\Gamma(\mu)} \right] = 
\ex[e^{\lambda X}].
\end{align*}

$(a)$ follows because the Gamma function is monotonically increasing in the range $[1.5, \infty)$. $(b)$ follows because $\mu = \frac{1}{\Gamma(1/p)} \geq 1/p$ for $p \geq 1$, and because for positive integers $r$, $\frac{\Gamma(x + r)}{\Gamma(x)} = \prod_{i=0}^{r-1} (x+i)$ is monotonically increasing in $x$. Since $X \in \Gamma^+(\mu, 1)$ and $X, Y$ have the same mean, we have that $Y \in \Gamma^+(\mu, 1)$ as well. 
\end{proof}
\begin{lemma}
Let $Y = GGamma(\frac{1}{p-1}, \frac{p}{p-1})$ for $p \geq 3$. Then, for $\mu = \ex[Y] = \frac{1}{\Gamma(1/p)}$, we have $Y \in \Gamma^-(\mu, 3/2)$.
\end{lemma}
\begin{proof}
Similarly to the previous lemma, we have for all $0 \leq \lambda \leq 2/3$:

\begin{align*}
\ex[e^{-\lambda Y}]&=
1 - \lambda \mu + \sum_{r=2}^\infty \left[\frac{(-\lambda)^r}{r!} \frac{\Gamma(\frac{1}{p} + \frac{r(p-1)}{p})}{\Gamma(\frac{1}{p})} \right]\\
&=1 - \lambda \mu + \sum_{r=1}^\infty \left[ \frac{\lambda^{2r}}{(2r)!} \cdot \frac{\Gamma(\frac{1}{p} + 2r \frac{p-1}{p})}{\Gamma(\frac{1}{p})} \left(1 - \frac{\lambda}{2r+1} \cdot \frac{\Gamma(\frac{1}{p} + (2r+1) \frac{p-1}{p})}{\Gamma(\frac{1}{p} + 2r \frac{p-1}{p})} \right)\right] \\
&=1 - \lambda \mu + \sum_{r=1}^\infty \left[ \frac{\lambda^{2r}}{(2r)!} \cdot \frac{\Gamma(\frac{1}{p} + 2r)}{\Gamma(\frac{1}{p})} \left(\frac{\Gamma(\frac{1}{p} + 2r\frac{p-1}{p})}{\Gamma(\frac{1}{p} + 2r)} - \frac{\lambda}{2r+1} \cdot \frac{\Gamma(\frac{1}{p} + (2r+1) \frac{p-1}{p})}{\Gamma(\frac{1}{p} + 2r)} \right)\right] \\
&\stackrel{(c)}{\leq}1 - \lambda \mu + \sum_{r=1}^\infty \left[ \frac{\lambda^{2r}}{(2r)!} \cdot \frac{\Gamma(\frac{1}{p} + 2r)}{\Gamma(\frac{1}{p})} \left(1 - \frac{\lambda}{2r+1} \cdot \frac{\Gamma(\frac{1}{p} + 2r+1)}{\Gamma(\frac{1}{p} + 2r)} \right)\right] \\
& \stackrel{(d)}{\leq}1 - \lambda \mu + \sum_{r=1}^\infty \left[ \frac{\lambda^{2r}}{(2r)!} \cdot \frac{\Gamma(\mu + 2r)}{\Gamma(\mu)} \left(1 - \frac{\lambda}{2r+1} \cdot \frac{\Gamma(\mu + 2r+1)}{\Gamma(\mu + 2r)} \right)\right] \\
& =1 - \lambda \mu + \sum_{r=2}^\infty \left[ \frac{(-\lambda)^r}{r!} \cdot \frac{\Gamma(\mu + r)}{\Gamma(\mu)}\right] = \ex[e^{-\lambda X}].
\end{align*}

Which, up to proving $(c), (d)$ hold, shows that $Y \in \Gamma^-(\mu, 3/2)$ since $X$ and $Y$ have the same mean and $X \in \Gamma^-(\mu, 0) \subset \Gamma^-(\mu, 3/2)$. $(c)$ follows because the change in each term in the sum is 

$$\frac{\lambda^{2r}}{(2r)!} \frac{1}{\Gamma(\frac{1}{p})} \left[\Gamma(\frac{1}{p} + 2r) - \Gamma(\frac{1}{p}+2r\frac{p-1}{p}) - \frac{\lambda}{2r+1}\left(\Gamma(\frac{1}{p}+2r+1) - \Gamma(\frac{1}{p}+(2r+1)\frac{p-1}{p}) \right) \right].$$

To show this expression is non-negative, it suffices to show that just the term in the brackets is positive, or equivalently, for all $r \geq 2, p \geq 3$:

$$\Gamma(\frac{1}{p} + 2r)\left(1 - \frac{\Gamma(\frac{1}{p} + 2r\frac{(p-1)}{p})}{\Gamma(\frac{1}{p} + 2r)}\right) \geq \frac{\lambda}{2r+1}\Gamma(\frac{1}{p}+2r+1)\left(1- \frac{\Gamma(\frac{1}{p}+(2r+1)\frac{p-1}{p})}{\Gamma(\frac{1}{p}+2r+1)} \right) .$$

Since we have $\Gamma(\frac{1}{p} + 2r + 1) = (\frac{1}{p} + 2r)\Gamma(\frac{1}{p} + 2r) \leq (2r + 1) (\frac{1}{p} + 2r)$, it further suffices to just show:

$$f(r, p) := \frac{\left(1 - \frac{\Gamma(\frac{1}{p} + 2r\frac{(p-1)}{p})}{\Gamma(\frac{1}{p} + 2r)}\right)}{\left(1- \frac{\Gamma(\frac{1}{p}+(2r+1)\frac{p-1}{p})}{\Gamma(\frac{1}{p}+2r+1)} \right)} \geq \lambda.$$

For any fixed $r \geq 2$, one can verify analytically that $f(r, p)$ is monotonically decreasing in $p$ over $p \in [1, \infty)$ and the limit as $p$ goes to infinity is $g(r) := \frac{2r \psi (2r)}{(2r+1) \psi (2r+1)}$ where $\psi$ is the digamma function $\psi(x) = \frac{\frac{\D}{\D{x}}\Gamma(x)}{\Gamma(x)}$. One can also verify analytically that $g(r)$ is monotonically increasing, and $g(2) \approx .6672$. So, for all $r \geq 2, p \geq 3$ we have $f(r, p) > 2/3$ and thus for $\lambda \in [0, 2/3]$, the inequality $(c)$ is satisfied.

$(d)$ follows by looking at the function

$$z(x) = \frac{\Gamma(x + r)}{\Gamma(x)} \left(1 - \frac{\lambda}{r+1} \cdot \frac{\Gamma(x + r + 1)}{\Gamma(x + r)}\right) = \left(1 - \frac{\lambda (x+r)}{r+1} \right)\prod_{i=0}^{r-1} (x+i).$$

For $r \geq 2, \lambda \leq 1$, one can verify analytically that $z(x)$ is monotonically increasing in the interval $(0, 1/2] \supseteq (0, \frac{1.2}{p}] \supseteq (0, \mu]$. Since $\mu \geq \frac{1}{p}$, this gives that each term in the right-hand-side of $(d)$ is larger than the corresponding term on the left-hand-side.
\end{proof}

\end{document}